\RequirePackage{fix-cm}
\documentclass[smallextended]{svjour3}       
\smartqed  
\usepackage{graphicx}
\input{epsf}
\begin{document}

\title{A simpler and more efficient algorithm for the next-to-shortest path problem
}


\author{Bang Ye Wu}


\institute{National Chung Cheng University, ChiaYi, Taiwan 621,
R.O.C.\\
\email{bangye@cs.ccu.edu.tw}
}

\date{Received: date / Accepted: date}

\maketitle

\begin{abstract}
Given an undirected graph $G=(V,E)$ with positive edge lengths and two vertices $s$ and $t$, the next-to-shortest path problem is to find an $st$-path which length is minimum amongst all $st$-paths strictly longer than the shortest path length. In this paper we show that the problem can be solved in linear time if the distances from $s$ and $t$ to all other vertices are given. Particularly our new algorithm runs in $O(|V|\log |V|+|E|)$ time for general graphs, which improves the previous result of $O(|V|^2)$ time for sparse graphs, and takes only linear time for unweighted graphs, planar graphs, and graphs with positive integer edge lengths.  

\end{abstract}

\section{Introduction}

Let $G=(V,E,w)$ be an undirected graph, in which $w$ is a positive edge length function.
For $s,t\in V$, an $st$-path is a simple path from $s$ to $t$, in which ``simple'' means there is no repeated vertex in the path. In this paper, a path always means a simple path.
The length of a path is the total length of all edges in the path.
An $st$-path is a shortest $st$-path if its length is minimum amongst all possible $st$-paths.
The shortest path length from $s$ to $t$ is denoted by $d(s,t)$ which is the length of their shortest path.
A \emph{next-to-shortest} $st$-path is an $st$-path which length is minimum amongst those the path lengths \emph{strictly larger} than $d(s,t)$. And the next-to-shortest path problem is to find a next-to-shortest $st$-path for given $G$, $s$ and $t$. 

While the shortest path problem has been widely studied and efficient algorithms have been proposed, the next-to-shortest path problem attracts researchers just in the last decade.
The problem was first studied by Lalgudi and Papaefthymiou in the directed version with no restriction to positive edge weight \cite{lal97}. They showed that the problem is intractable for path and can be efficiently solved for walk (allowing repeated vertices).
Algorithms for the problem on special graphs were also studied \cite{bar07,mod06}.
The first polynomial algorithm for undirected positive version, i.e., the next-to-shortest path defined in this paper,
was developed by Krasikov and Noble, and their algorithm takes $O(n^3m)$ time \cite{kra04}, in which $n$ and $m$ are the number of vertices and edges, respectively.
The time complexity was then reduced to $O(n^3)$ by Li et al. \cite{li06}.
Recently, Kao et al. further improved the time complexity to $O(n^2)$ \cite{kao10}.
In this paper, we show that the problem can be solved in linear time if the distances from $s$ and $t$ to all other vertices are given.

Let $D$ be the union of all shortest $st$-paths. 
For convenience let $D^+$ be the digraph obtained from $D$ by orientating all edges toward $t$.
Apparently $s$ and $t$ are in $V(D^+)$ and, for any $x,y\in V(D^+)$, any (directed) $xy$-path in $D^+$ is a shortest $xy$-path (undirected) in $G$.
An outward subpath of an $st$-path is a path consisting of edges in $E-E(D)$ and the both endpoints are in $V(D)$; and a backward subpath is a maximal subpath using edges in $E(D^+)$ but with reverse direction.
Since a next-to-shortest path either contains an edge in $E-E(D)$ or not, we divide the problem into two subproblems, and the better of the solutions of the two subproblems is the optimal path.
The first subproblem looks for a shortest path using at least one edge not in $E(D)$, and the second subproblem looks for a shortest path consisting of only edges in $E(D)$ but with length larger than $d(s,t)$. Following the previous names in \cite{kao10}, we name the optimal paths of the first and the second subproblems as ``optimal outward path'' (optimal path with an outward subpath) and ``optimal backward path'' (optimal path with a backward subpath), respectively. 
Since any $st$-path in $D^+$ has length $d(s,t)$, the optimal backward path uses at least one edge with reverse direction. By the optimality the following result was shown in \cite{kao10}
and it is the basis of the algorithms in the previous and this papers.
\begin{lemma}\label{c1}
The optimal outward path contains exactly one outward subpath and no backward subpath. The optimal backward path contains exactly one backward subpath.
\end{lemma}

The reason for the two observations is the same: If there are two non-consecutive backward or outward subpaths, we can replace one of them with a subpath in $D^+$ to obtain a better one. 
Due to \cite{kao10}, the optimal outward path subproblem can be solved in $O(m+n\log n)$ time. 
But, for the optimal backward subproblem, they only gave an $O(n^2)$ time algorithm.
The contribution of this paper is as follows.
\begin{itemize}
\item We give an $O(m+n\log n)$ time algorithm for the optimal backward subproblem, which also reduces the total time complexity of the whole algorithm for sparse graphs. 
\item We give an algorithm for finding an optimal outward path. The time complexity is the same as Kao's algorithm for general graphs but the new algorithm is simpler and avoids the sorting step in Kao's algorithm.  
\item More precisely, if the distances from $s$ and $t$ to all other vertices are given, both our new algorithms take only linear time.
That is, for graphs on which the \emph{single source shortest paths} (SSSP) problem can be solved in $t(m,n)$ time, the next-to-shortest path problem can be solved in $O(t(m,n)+m+n)$ time. Consequently the next-to-shortest path problem can be solved in linear time for undirected unweighted graph, for undirected planar graph with positive edge weights, and for undirected graph with positive integral edge weights.
\end{itemize}

The remaining sections are organized as follows. We introduce notations and derive some basic properties in Section 2. The algorithm for the optimal backward path is shown in Section 3, and in Section 4, we give a simpler algorithm for the outward path problem. Finally concluding remarks are given in Section 5. 
 
\section{Preliminaries}

Throughout this paper, we shall assume that $(G,s,t)$ is the instance of the problem, in which $G=(V,E,w)$ is the input graph with vertex set $V$, edge set $E$ and edge length $w$. 
$s$ and $t$ are two vertices in $V$. The graph $G$ is simple, connected and undirected, and the edge lengths are all positive. We shall use $n=|V|$ and $m=|E|$.

For a graph $H$, $V(H)$ and $E(H)$ denote the vertex and edge sets, respectively. 
For two vertices $x$ and $y$ on a path $P$, let $P[x,y]$ denote the subpath from $x$ to $y$ and $w(P)=\sum_{e\in E(P)}w(e)$ denote the length of the path.
Let $d(x,y)$ denote the shortest path length from $x$ to $y$ in $G$, which is also called as the distance from $x$ to $y$. 
For convenience, let $d_s(v)=d(s,v)$ and $d_t(v)=d(v,t)$. 

To show the time complexities more precisely, we shall assume the distances from $s$ and $t$ to all other vertices are given.
These distances can be found by solving the single source shortest paths (SSSP) problem. For general undirected and positive weight graphs (the most general setting of the problem discussed in this paper), the SSSP problem can be solved in $O(m+n\log n)$ time \cite{cor01,dijk}, and more efficient algorithms exist for special graphs or graphs with restrictions on edge lengths.
 
As defined in the introduction, let $D$ be the union of all shortest $st$-paths and $D^+$ is obtained from $D$ by orientating all edges toward $t$. Constructing $D$ and $D^+$ can be done in linear time as follows. A vertex $v$ is in $V(D)$ iff $d_s(v)+d_t(v)=d(s,t)$, and, for both $u$ and $v$ in $V(D)$, a directed edge $(u,v)\in E(D^+)$ iff $d_s(v)=d_s(u)+w(u,v)$. Similarly a shortest path tree rooted at $s$ can also be constructed in linear time if the distances from $s$ to all others are given.

Since all edge lengths are positive, $D^+$ is a directed acyclic graph (dag) and we may use the terms such as parent, child, ancestor and descendant as in a rooted tree. Also, for convenience, we abuse the notation $d(x,y)$ for the distance from $x$ to $y$ in $D$ and $D^+$  as long as $x$ is an ancestor of $y$.   
We shall use \emph{immediate dominators} in our algorithm. 
A vertex $v\in V(D^+)$ is an $s$-dominator of another vertex $x$ iff all paths from $s$ to $x$ contain $v$. 
An $s$-dominator $v$ is an $s$-immediate-dominator of $x$, denoted by $I_s(x)$, if it is the one closest to $x$, i.e., any other $s$-dominator of $x$ is an $s$-dominator of $v$.
We remind that, for any vertex $x\in V(D^+)$, there exist a path from $s$ to $x$ and also a path from $x$ to $t$.
Apparently any vertex has a unique $s$-immediate-dominator and all $s$-dominators, as well as the $s$-immediate-dominator, are ancestors of the vertex.
Similarly we define the $t$-dominator, i.e., $v$ is a $t$-dominator of $x$ iff any $xt$-path contains $v$, and $I_t(x)$ is the $t$-dominator closest to $x$.
Note that $s$ is an $s$-dominator and $t$ is a $t$-dominator for any other vertex in $V(D)$.  

Finding immediate dominator is one of the most fundamental problems in the area of global flow analysis and 
program optimization. The first algorithm for the problem was proposed in 1969, and then had
been improved several times. A linear time algorithm for finding the immediate dominator for each vertex was given in \cite{als99}. 

We define a binary relation on $V(D^+)$: $x\prec y$ iff $x$ is an ancestor of $y$. In our definition, a vertex is not an ancestor of itself. Also define $x\preceq y$ iff $x$ is an ancestor of $y$ or $x=y$.
We derive some properties used in this paper.

\begin{lemma}\label{idom1}
For any vertices $x$ and $y$ in $D^+$ and $y\prec x$, either $I_s(x)\preceq y$ or $y\prec I_s(x)$.
Similarly either $I_t(y)\preceq x$ or $x\prec I_t(y)$.
\end{lemma}
\begin{proof}
We show the first statement and the second statement is similar.
If neither of the two conditions holds, there is an $sx$-path passing through $y$ and avoiding $I_s(x)$, which contradicts the definition of dominator.
\qed\end{proof}

The following corollary comes from Lemma~\ref{idom1}.
\begin{corollary}\label{idom2}
If $y\prec x$ and $d_s(I_s(x))<d_s(y)$, then $I_s(x)\prec y$.
Similarly,
if $y\prec x$ and $d_s(I_t(y))>d_s(x)$, then $x\prec I_t(y)$.
\end{corollary}

\section{Optimal backward path}
In this section we show an efficient algorithm for finding an optimal backward path.
For this problem, only vertices and edges in $D^+$ need considering and any vertex is assumed in $D^+$ in this section.
Since the numbers of vertices and edges of $D^+$ are also bounded by $n$ and $m$ respectively, we shall neither distinguish $|V|$ and $|V(D^+)|$, nor $|E|$ and $|E(D^+)|$.

\subsection{The objective function and the constraints} 

By the previous result shown in the introduction, an optimal backward path has the form $Q_1\circ Q_2^{-1}\circ Q_3$, in which ``$\circ$'' means concatenation, $Q_i$ are paths in $D^+$ and $Q_2^{-1}$ means the reverse path of $Q_2$. Since the optimal path is required to be simple, the three subpaths must be simple and internally disjoint, in which two paths are internally disjoint if they have no common vertex except for their endpoints. Therefore our goal is to find $x,y\in V(D)$ minimizing 
\begin{eqnarray}
d(s,x)+d(x,y)+d(y,t)
\end{eqnarray}
subject to that there are an $sx$-path, a $yx$-path and a $yt$-path in $D^+$ which are mutually internally disjoint. 
If $x$ and $y$ satisfy the constraint, we say $(x,y)$ is valid. 

Since all paths in $D^+$ are shortest, we have $d(y,t)=d(y,x)+d(x,t)$ and $d(s,x)+d(x,t)=d(s,t)$, and the objective function can be simplified to $d(s,t)+2d(x,y)$ and also equivalent to $d(x,y)$ since $d(s,t)$ is independent on $x$ and $y$.

For any vertex $x$, let $C(x)=\{v| I_s(x)\prec v\prec x\}$.
The vertices in $C(x)$ form a closed region in the sense that no path can enter this region without passing through $I_s(x)$, and it is easy to see that, for any vertex $x\neq s$, $C(x)=\emptyset$ iff $x$ has only one parent.
The most important thing is that, as shown later, the vertices which are valid for $x$ must be in $C(x)$. 
\begin{lemma}\label{hier}
For any $y\in C(x)$, $I_s(x)\preceq I_s(y)$.
\end{lemma}
\begin{proof}
Since $I_s(x)\prec y$, we have $I_s(x)\preceq I_s(y)$ or $I_s(y)\prec I_s(x)$ by Lemma~\ref{idom1}.
By the definition of $I_s(x)$, the in-neighbors of $C(x)$ are contained in $C(x)\cup \{I_s(x)\}$.
Therefore it is impossible that $I_s(y)\prec I_s(x)$.
\qed\end{proof}

\begin{lemma}\label{close}
If $y\in C(x)$, there are two internally disjoint paths from $I_s(x)$, and $y$ respectively, to $x$. 
\end{lemma}
\begin{proof}
Let $p=I_s(x)$.
By the definition of immediate dominator, no vertex in $C(x)$ is a $px$-cut and therefore there are two internally disjoint $px$-paths, said $P_1$ and $P_2$.
If $y$ is on one of them, we have done.
Otherwise, let $P_3$ be any $yx$-path and $q$ be the first vertex on $P_3$ and also in $V(P_1)\cup V(P_2)$.
W.l.o.g. let $q\in V(P_1)$. 
Then, the path $P_3[y,q]\circ P_1[q,x]$ is a $yx$-path disjoint to $P_2$.

\qed\end{proof}

Next we derive the objective function and its constraints.
\begin{lemma}\label{cand}
If the pair $(x,y)$ is valid, then $y\in C(x)$ and $x\prec I_t(y)$.
\end{lemma}
\begin{proof}
By definition, $y\prec x$. By Lemma~\ref{idom1}, either $I_s(x)\prec y$ or $y\preceq I_s(x)$. 
If $y\preceq I_s(x)$, by the definition of immediate dominator, any $sx$-path and $yx$-path contain $I_s(x)$ simultaneously and cannot be disjoint.
Therefore we have $I_s(x)\prec y$.
The relation $x\prec I_t(y)$ can be shown similarly.
\qed\end{proof}
Define 
\begin{eqnarray}
g(x,y)=\left\{\begin{array}{ll}
d(y,x)\;&\mbox{if }y\in C(x)\mbox{ and }x\prec I_t(y) \\
\infty&\mbox{otherwise}
\end{array}\right. \label{obj}
\end{eqnarray}
and let $g^*(x)=\min_y g(x,y)$.

\begin{figure}[t]
\begin{center}
\epsfbox{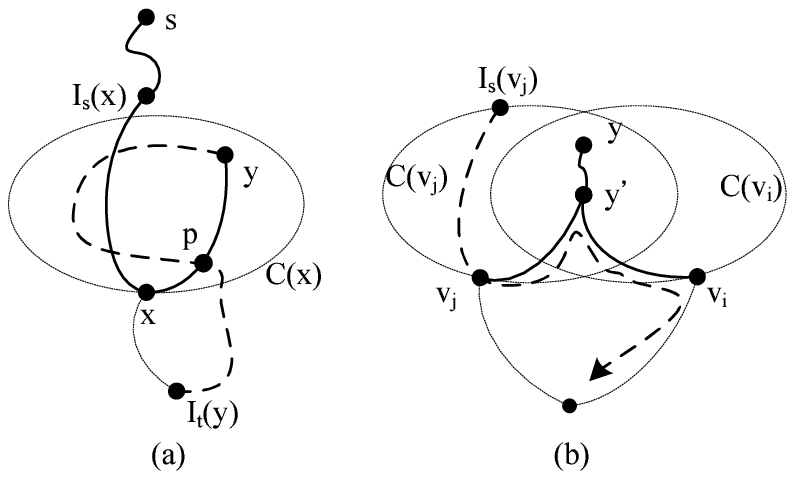}
\caption{(a). For Lemma \ref{suff}: $P_1$ is the solid line and $P_2$ is the dash line.
(b). For Corollary~\ref{black}: The dash line illustrates a feasible backward path.}
\label{backfig}
\end{center}
\end{figure}
\begin{lemma}\label{suff}
If $g^*(x)\neq \infty$ and $y^*=\arg\min_y g(x,y)$, then $(x,y^*)$ is valid.
\end{lemma}
\begin{proof}
By Lemma~\ref{close}, since $y^*\in C(x)$, there are two disjoint paths from $I_s(x)$, and $y^*$ respectively, to $x$. 
Therefore we have a simple path, said $P_1$, from $s$, passing through $I_s(x)$ to $x$, and then from $x$ to $y^*$ by backward edges.
Since $x\prec I_t(y^*)$, there must be a path, said $P_2$, from $y^*$ to $t$ and avoiding $x$.
We shall show that $P_1$ and $P_2$ are disjoint, which completes the proof.
Suppose to the contrary that $p\neq y^*$ is the last common vertex of $P_1$ and $P_2$, i.e., any other common vertex
precedes $p$ in $P_2$ (Figure~\ref{backfig}.(a)). 
Since $p\in V(P_2)$, we have $y^*\prec p$, and $p\prec x$ because $p\in V(P_1)$. So, we have $p\in C(x)$.
Since $P_2[p,I_t(y)]$ is a path avoiding $x$ and $p\prec x\prec I_t(y^*)$, we have $x\prec I_t(p)$.
Therefore $g(x,p)=d(p,x)\neq\infty$ and $d(p,x)<d(y^*,x)$ since $y^*$ is an ancestor of $p$, 
a contradiction to the optimality of $y^*$.
That is, $P_1$ and $P_2$ must be disjoint.
\qed
\end{proof}
The proof of Lemma~\ref{suff} is constructive, and it implies an algorithm for finding a corresponding backward path for given $x$ and $y^*$. Furthermore the time complexity is apparently linear.
For the simplicity, in the following, we only focus on finding the length of the optimal backward path.
By Lemmas~\ref{cand} and \ref{suff}, our goal is to find $x$ and $y$ minimizing $g$, i.e., 
\begin{eqnarray}
{\rm OPT} =\min_x\min_{y\in C(x)}\{d_s(x)-d_s(y)| x\prec I_t(y) \} \label{obj1}
\end{eqnarray}
Or, by Corollary~\ref{idom2}, it can be also written as 
\begin{eqnarray}
{\rm OPT} =\min_x\min_{y\in C(x)}\{d_s(x)-d_s(y)| d_s(x)< d_s(I_t(y)) \} \label{obj2}
\end{eqnarray}
The convenience of the latter form is that we can easily determine the ancestor relation by simply comparing their $d_s$ values.
The above formula provides us a way to find the optimal backward path: for each vertex $x$, checking each vertex $y\in C(x)$. But the naive method takes at least $\Theta(n^2)$ time in worst case since $|C(x)|$ may be linear in $n$.

\subsection{An efficient algorithm}
We say ``the pair $(x,y)$ is feasible'' or ``$y$ is feasible for $x$'' if $g(x,y)\neq\infty$. We also say ``$x$ is feasible'' if there exists $y$ which is feasible for $x$.
Note that a feasible $(x,y)$ may be not valid. However, our algorithm find the $(x,y)$ minimizing function $g$, and by Lemma~\ref{suff} it must be valid.

Let $\mathcal{A}(x)$ be the set of parents of $x$.
Our algorithm basically finds $g^*(x)$ for each $x$ according to the following formula: 
\begin{eqnarray}
g^*(x)=\min_{p\in \mathcal{A}(x)}\min_{y\preceq p}g(x,y), \label{allp}
\end{eqnarray}
and ${\rm OPT} =\min_x\{g^*(x)\}$.
We denote by $F$ the set of all the feasible vertices, i.e.,
\[ F=\{x|g^*(x)\neq \infty\} \]

To make the algorithm efficient, we derive some properties to avoid non-necessary searches.
\begin{lemma}\label{fp}
If $y\prec x$ and $y\in F$, then $g(x,u)>g^*(y)$ for any $u\prec y$.
\end{lemma}
\begin{proof}
If $y\preceq I_s(x)$, $\min_{u\prec y}\{g(x,u)\}=\infty$ and the result holds since $y\in F$.
We only need to consider the remaining case that $I_s(x)\prec y$.
Let $u\prec y$ and $g(x,u)\neq \infty$.
By definition, $u\in C(x)$ and $x\prec I_t(u)$.
Since $y\in C(x)$, by Lemma~\ref{hier}, $I_s(x)\preceq I_s(y)$. 
If $I_s(y)\prec u$, $u$ is also feasible for $y$ and $d(u,y)<d(u,x)=d(u,y)+d(y,x)$.
Otherwise $d(u,x)> d(I_s(y),y)$. Since $y\in F$, by definition $g^*(y)<d(I_s(y),y)$.
\qed\end{proof}
\begin{lemma}\label{nfp}
If $y\prec x$ and $y\notin F$, then $g(x,u)=\infty$ for any $u\in C(y)$.
\end{lemma}
\begin{proof}
Since $y\notin F$, for any vertex $u\in C(y)$, i.e., $I_s(y)\prec u\prec y$, we have $I_t(u)\preceq y$.
Since $y\prec x$,  $I_t(u)\prec x$ and $u$ cannot be feasible for $x$. 
\qed\end{proof}

\begin{lemma}\label{desf}
Let $v_j$ and $v_i$ be two descendants of $y$ and $d_s(v_j)\leq d_s(v_i)$.
If $g^*(v_j)\neq \infty$,  
$\min_{u\prec y}\{g(v_i,u)\}\geq g^*(v_j)$. 
\end{lemma}
\begin{proof}
If there exists $u$ such that $I_s(v_j)\prec u\prec y$ and $g(v_i,u)\neq \infty$, then $d_s(I_t(u))>d_s(v_i)\geq d_s(v_j)$. Since $u$ is also an ancestor of $v_j$, $v_j\prec I_t(u)$ by Corollary~\ref{idom2}. Therefore $g(v_i,u)=d_s(v_i)-d_s(u)\geq d_s(v_j)-d_s(u)\geq g^*(v_j)$.
For otherwise there is no such $u$, and we have $\min_{x\prec y}\{g(v_i,x)\}\geq d_s(v_i)-d_s(I_s(v_j))\geq d_s(v_j)-d_s(I_s(v_j))> g^*(v_j)$.
\qed\end{proof}

\begin{corollary}\label{black}
Let $v_j$ and $v_i$ be two vertices and $d_s(v_j)\leq d_s(v_i)$. If $y\in C(v_j)\cap C(v_i)$ and $v_j$ is not an ancestor of $v_i$, then $g^*(v_j)\neq \infty$ and $\min_{u\prec y}\{g(v_i,u)\}\geq g^*(v_j)$.
\end{corollary}
\begin{proof}
Since $y\in C(v_j)\cap C(v_i)$, $y$ is a common ancestor of $v_j$ and $v_i$.
Let $y'$ be a lowest common ancestor of them and $y\preceq y'$.
Since $v_j$ is not an ancestor of $v_i$, we have $y'\neq v_j$ and there is a path from $y'$ to $v_i$ avoiding $v_j$.
By definition $v_j\prec I_t(y')$ and $g^*(v_j)\neq \infty$ (Figure~\ref{backfig}.(b)).
The inequality follows directly from Lemma~\ref{desf}.
\qed\end{proof}

Our algorithm for the optimal backward path is as follows. 
\begin{tabbing}
\hspace*{1.5em} \= \hspace*{1.2em} \= \hspace*{1.2em} \= \hspace*{1.2em} \= \hspace*{1.2em} \= \kill \\
{\bf Algorithm} Bk\_N2SP \\
{\bf Input: }The digraph $D^+$. \\
{\bf Output: }The length of the optimal backward path. \\
{\bf 1:}\>find a topological order of $D^+$ and label the vertices such that \\
\>if $(v_i,v_j)\in E(D^+)$, $i<j$;\\
{\bf 2:}\>find $I_s(v)$ and $I_t(v)$ for each $v$;\\
{\bf 3:}\>$\beta\leftarrow \infty$;\\
\>$color(v)\leftarrow white$, $\forall v\in V(D^+)$;\\
\>$color(s)\leftarrow black$;\\
{\bf 4:}\>for $i\leftarrow 2$ to $n-1$ do \\
{\bf 5:}\>\>for each parent $p$ of $v_i$ do \\
{\bf 6:}\>\>\>$y\leftarrow p$; \\
{\bf 7:}\>\>\>while $I_s(v_i)\prec y$ and $g(v_i,y)=\infty$ and $color(y)=white$ do\\
{\bf 8:}\>\>\>\>$color(y)\leftarrow black$; $y\leftarrow I_s(y)$; \\
{\bf 9:}\>\>\>if $g(v_i,y)\neq \infty$ then\\
{\bf 10:}\>\>\>\>$\beta\leftarrow \min\{\beta,d_s(v_i)-d_s(y)\}$; \\
\>\>\>\>$color(v_i)\leftarrow black$; $color(y)\leftarrow black$;\\
{\bf 11:}\>\>end for next parent;\\
{\bf 12:}\>end for next $i$;\\
{\bf 13:}\>output $d(s,t)+2\beta$\\
\end{tabbing}

In the algorithm, each vertex $v$ is associated with a color, which is white initially and may be set to black as the algorithm runs.
The algorithm begins with a preprocessing stage at Steps 1--3.
We first arrange the vertices according to a topological order in $D^+$.
Note that $s=v_1$ and $t=v_n$.
Then we find the $s$- and $t$-immediate dominators for each vertex.
All vertices are assigned white color except that $s$ is colored black.
The variable $\beta$ is used to keep the objective value of the best solution found so far.
In the main loop from Steps 4 to 12, we deal with all the vertices one by one except for $s$ and $t$. 
In the $i$-th iteration, we try to find any feasible $y\in C(v_i)$ for $v_i$ from each parent of $v_i$ (Steps 5--11).

\subsection{Correctness and time complexity}
We shall show the correctness of the algorithm by examining the feasibility and the optimality.
\subsubsection*{Feasibility.}
The algorithm finds solutions only at Step 10. By Lemma~\ref{suff}, the final solution is feasible as long as its minimality can be ensured. 

\subsubsection*{Optimality.}
Apparently neither $s$ nor $t$ can be a feasible vertex.
By Eq. (\ref{allp}) what we need to show is that the solutions we skipped are really not better.
By Eq. (\ref{obj1}), we should check all $y\in C(v_i)$ for each $v_i$.
There are two kinds of solutions skipped by the algorithm.
\begin{itemize}
\item {\bf Type 1:}
The first kind of possible solutions ignored by the algorithm is at Step 8 where we look for feasible solution $g(v_i,y)$ for any $y\preceq p$ but do not try all such $y$. Instead, we jump to $I_s(y)$ after checking $y$ and skip the vertices in $C(y)$.
If $y$ is feasible, by Lemma~\ref{fp}, we do not need to check $g(v_i,u)$ for any ancestor $u$ of $y$; and if 
$y$ is not feasible, by Lemma~\ref{nfp}, ignoring $(v_i,u)$ for any $u\in C(y)$ does not affect the optimality.

\item {\bf Type 2:} The second kind of skipped solutions is due to the conditions of the while-loop at Step 7.
The while-loop stops when $y\preceq I_s(v_i)$ or $g(v_i,y)\neq \infty$ or $color(y)=black$.
Except for the first condition, the loop may terminate before reaching $I_s(v_i)$.
For the second condition, if $g(v_i,y)\neq \infty$, the optimality is ensured by Lemma~\ref{fp}; and 
we divide the third condition $color(y)=black$ into two sub-cases according to when it turns black. 
\begin{itemize}
\item If $y$ is colored black in this iteration, $y$ must have been checked at this iteration from another parent of $v_i$. Therefore, if there exists any ancestor of $y$ we need checking, it must have already been checked from that parent. 
\item Otherwise $y$ is colored black before the $i$-th iteration, and this implies that $y\in C(v_j)$ for some $j<i$ or $y$ is feasible (marked black at Step 10 in the iteration checking $y$ with its ancestor). If $y$ is feasible, we can safely skip any ancestor of $y$ by Lemma~\ref{fp}. Otherwise, if $v_j$ is not an ancestor of $v_i$, we have that $v_i$ or $v_j$ must be feasible and $\min_{u\prec y}\{g(v_i,u)\}\geq \min\{g^*(v_i),g^*(v_j)\}$ by Corollary~\ref{black}. Note that we still need to and do check $g(v_i,y)$ at Step 9 because $d_s(I_s(v_i))$ may be smaller than $d_s(I_s(v_j))$. 

The remaining case is that $v_j\prec v_i$ and $y\in C(v_j)$.
Let $(y_1, y_2,\ldots ,y_k=y)$ be the sequence of vertices checked at the while-loop.
Since $y_{q+1}=I_s(y_q)$ for $1\leq q\leq k-1$ and $I_s(v_j)\prec y_k\prec v_j$, we have that $v_j$ does not appear in the sequence. Hence, there exists $y_q$ which is in $C(v_j)$ and has a path to $v_i$ avoiding $v_j$.
Therefore $v_j$ is feasible, 
and then similar to Corollary~\ref{black}, it is not necessary to check any ancestor of $y$ for $v_i$.  
\end{itemize}
\end{itemize}
By the above explanation, we conclude the correctness of the algorithm. 

\begin{lemma}\label{correct}
Algorithm Bk\_N2SP computes the optimal backward path length correctly.
\end{lemma}

\subsubsection*{Time complexity.}
We show that the time complexity of our algorithm is $O(n+m)$.
Step 1 takes linear time for finding a topological order and Step 2 takes $O(m)$ time \cite{als99}.
Step 3 takes $O(n)$ time.
The inner loop (Steps 6--10) is entered $O(m)$ times since the total number of parents of all vertices is bounded by the number of edges. Since $y$ is always an ancestor of $v_i$, the conditions ``$I_s(v_i)\prec y$'' and ``$g(v_i,y)=\infty$'', as well as to compute $g(v_i,y)$, can all be done in constant time by checking the $d_s$ values. 
The remaining question is how many times Step 8 is executed.
By the condition of the while loop, only white vertices will be colored black at Step 8, and therefore it is executed at most $n$ times in total.

The algorithm assumes that $D^+$ is the input. 
It does not matter since $D^+$ can be constructed in linear time if the distances $d_s(v)$ and $d_t(v)$ are given for all $v$.
\begin{lemma}\label{backtime}
The time complexity of the algorithm Bk\_N2SP is $O(m+n)$.
\end{lemma}

By Lemmas~\ref{correct} and \ref{backtime}, we have the next result.
\begin{theorem}\label{thm:back}
The optimal backward path problem on undirected graphs with positive edge lengths can be solved in linear time if the distances from $s$ and $t$ to all the others are given.
\end{theorem}

For general undirected graph with positive edge weights, the SSSP problem can be solved in $O(n\log n+m)$ time by the Dijkstra's algorithm using a Fibonacci Heap \cite{dijk}.
Therefore the next corollary directly comes from Theorem~\ref{thm:back} and the result of the outward path in \cite{kao10}. 
\begin{corollary}\label{cor:back}
The next-to-shortest path problem on undirected graphs with positive edge lengths can be solved in $O(n\log n+m)$ time.
\end{corollary}

\section{Optimal outward path}

In this section we show an efficient algorithm for finding an optimal outward path. As described in the introduction, an optimal outward path contains exactly one outward subpath and has no backward subpath, in which an outward subpath is a path $P$ such that $E(P)\subset E-E(D)$ and both endpoints of $P$ are in $V(D)$.
An outward path must be strictly longer than a shortest path between its endpoints. Otherwise it should be entirely in $D$. Therefore the length of an outward $st$-path must be strictly larger than $d(s,t)$. Our goal is to find an minimum length $st$-path with an outward subpath. 

Let $T$ be any shortest-path tree of $G$ rooted at $s$ and $R=T-E(D)$ denote the graph obtained by removing edges in $E(T)\cap E(D)$ from $T$.
Apparently $R$ is a forest consisting of subtrees of $T$ and $V(R)=V(T)=V$.
By the definition of $D$, any shortest path between two vertices in $D$ must be included in $E(D)$. For any $v\in V(D)$, the path from $s$ to $v$ on $T$ must be entirely within $E(D)$ and therefore $v$ must be a root of a subtree of $R$.
Furthermore, the root of any subtree of $R$ must be in $V(D)$ because the edge between it and its parent is removed and we only remove edges between two vertices in $V(D)$.
Let $\widetilde{E}$ denote the set of edges $(x,y)$ such that $(x,y)\in E-E(T)\cup E(D)$ and $x$ and $y$ are in different subtrees of $R$. 
We show the next lemma.

\begin{lemma}\label{out1}
An optimal outward path $P$ contains one edge in $\widetilde{E}$.
\end{lemma}
\begin{proof}
By definition $P$ contains an outward subpath. Since the both endpoints of this outward subpath are in $V(D)$, they must be in different subtrees of $R$, and $P$ must have an edge in $\widetilde{E}$.
\qed 
\end{proof}

Define 
\begin{eqnarray} 
f(x,y)=\left\{\begin{array}{ll}
d_s(x)+w(x,y)+d_t(y)\;& \mbox{if }(x,y)\in \widetilde{E}\\
\infty&\mbox{otherwise}
\end{array}
\right.
\end{eqnarray}
Note that, since $G$ is undirected, both $(x,y)$ and $(y,x)$ denote the same edge. But $f(x,y)\neq f(y,x)$ in general. The following lemma is crucial for our result. 

\begin{figure}[t]
\begin{center}
\epsfbox{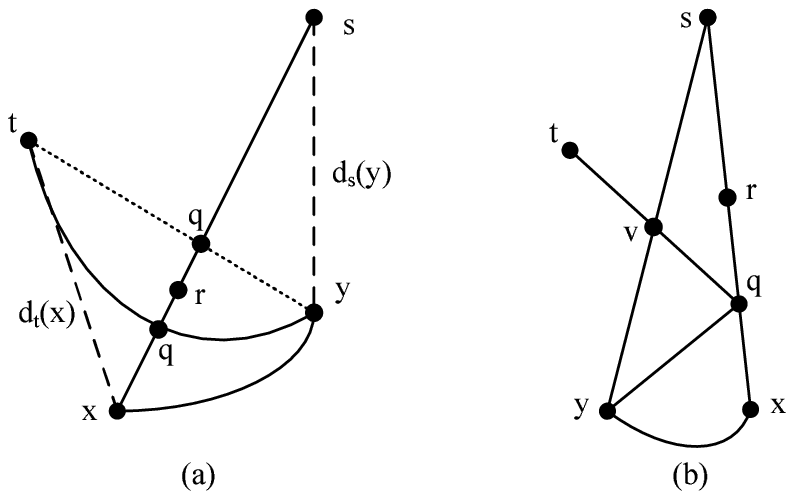}
\caption{The existence of a simple path with length $f(x,y)$. (a): Two cases of $P_2$; the dot line is impossible. (b): The case of $P_3$ intersecting $P_2[q,t]$.}
\label{outf}
\end{center}
\end{figure}

\begin{lemma}\label{outopt}
If $(x,y)$ minimizes function $f$ and $f(x,y)\neq \infty$, then there exists a simple $st$-path of length $f(x,y)$ and with one edge in $\widetilde{E}$. Such a path is an optimal outward path.  
\end{lemma}
\begin{proof}
We shall show the existence of such a simple path, and then it is an optimal outward path by Lemma~\ref{out1}.

Let $P_1$ be the shortest path from $s$ to $x$ on $T$ and $P_2$ any shortest path from $y$ to $t$. If $P_1$ and $P_2$ are disjoint, $P_1\circ (x,y)\circ P_2$ is a desired path since $(x,y)$ is an edge in $\widetilde{E}$.
Otherwise let $q\in V(P_1)\cap V(P_2)$.
By the triangle inequalities, seeing Figure~\ref{outf}.(a), we have 
\begin{eqnarray*}
f(y,x)&=&d_s(y)+w(y,x)+d_t(x)\\
&\leq&w(P_1[s,q]\circ P_2[q,y])+w(x,y)+w(P_1[x,q]\circ P_2[q,t])\\
&=&w(P_1)+w(x,y)+w(P_2)=f(x,y) 
\end{eqnarray*}
and the equality holds when 
\begin{eqnarray}
d_s(y)=w(P_1[s,q]\circ P_2[q,y])\label{outeq1}
\end{eqnarray}
and 
\begin{eqnarray}
d_t(x)=w(P_1[x,q]\circ P_2[q,t]) \label{outeq2}
\end{eqnarray}

Let $r$ be the root of the subtree which $x$ belongs to. If $q$ is an internal node of $P_1[s,r]$, i.e., $q$ is a ancestor of $r$ on $T$, since $r\in V(D)$, we have 
\begin{eqnarray*}
d_t(x)&\leq& w(P_1[x,r])+d_t(r)\\
&=&w(P_1[x,r])+d_t(q)-d(q,r)\\
&=&w(P_1[x,q]\circ P_2[q,t])-2d(q,r)\\
&<&w(P_1[x,q]\circ P_2[q,t]). 
\end{eqnarray*}
That is, the equality of Eq. (\ref{outeq2}) cannot hold and we have $f(y,x)<f(x,y)$, which contradicts to the minimality of $f(x,y)$.
Therefore $q$ is on $P_1[r,x]$.
In the following we assume that $q$ is the last vertex of $P_2$ which is also in $P_1$ if $P_2$ intersects $P_1$ at more than one vertices, i.e., $P_2[q,t]$ intersects $P_1$ only at $q$.
Let $P_3$ be the path from $s$ to $y$ on $T$. 
Since $x$ and $y$ are at different subtrees of $R$, the path $P_3$ is disjoint with the path $P_1[r,x]$. 
If $P_3$ is disjoint to $P_2[q,t]$, the path $P_3\circ (y,x)\circ P_1[x,q]\circ P_2[q,t]$ is a simple path and its length is $f(y,x)$ and also equals to $f(x,y)$ by Eq.~(\ref{outeq2}), seeing Figure~\ref{outf}.(a).
Otherwise, let $v$ be the vertex in $V(P_3)\cap V(P_2[q,t])$ and closest to $y$.
We shall show that $P_1$ is disjoint to $P_3[y,v]$, and then $P_1\circ (x,y)\circ P_3[y,v]\circ P_2[v,t]$ is a simple path (Figure~\ref{outf}.(b)) and its length is 
\[ d_s(x)+w(x,y)+w(P_3[y,v])+w(P_2[v,t]), \]
which is at most $f(x,y)$ since $w(P_3[y,v])=d(y,v)\leq w(P_2[y,v])$.

To see $P_1$ must be disjoint to $P_3[y,v]$, we first note that $P_3[y,v]$ is disjoint to $P_1[r,x]$ or otherwise $y$ and $x$ would be in the same subtree of $R$.
Next, similar to that $q$ cannot be an ancestor of $r$, since $v$ is also a vertex on $P_2$,  we have that $v$ cannot be an ancestor of $r$ on $T$. That is, $v$ must be below the lowest common ancestor of $x$ and $y$ on $T$, and therefore $P_3[y,v]$ is disjoint to $P_1$.  
\qed  
\end{proof}

\begin{theorem}\label{thm:out}
For an undirected graph with positive edge lengths, the optimal outward path problem can be solved in $O(m+n)$ time if $d_s(v)$ and $d_t(v)$ are given for all $v$.  
\end{theorem}
\begin{proof}
By Lemma~\ref{outopt}, the length of the optimal outward $st$-path is the minimum value of function $f$. To compute $(x,y)$ minimizing $f$, we first construct $D$ and a shortest path tree $T$, and then find the edge set $\widetilde{E}$. The minimum can be found by checking both $f(x,y)$ and $f(y,x)$ for all edges $(x,y)\in \widetilde{E}$. The time complexity is linear if the distances $d_s(v)$ and $d_t(v)$ for all $v$ are given.
Once $(x,y)$ is found, by the method in the proof of Lemma~\ref{outopt}, the corresponding path can be constructed in linear time.
\qed\end{proof}

By Theorems~\ref{thm:back} and \ref{thm:out}, we have the next result.
\begin{corollary}
For undirected graphs with positive edge lengths, if the single source shortest path problem can be solved in $O(t(m,n))$ time, the next-to-shortest path problem can be solved in $O(t(m,n)+m+n)$ time.
\end{corollary}
Important graph classes for which the single source shortest path problem can be solved in linear time include unweighted graphs (by BFS \cite{cor01}), planar graphs \cite{hen97}, and integral edge length graphs \cite{thr99}.

\section{Concluding remarks}

It is easy to show that the next-to-shortest path problem is at least as hard as the shortest path problem.
Given an instance of the shortest path problem, we add a dummy edge between $s$ and $t$ with sufficient small weight.
Then if there is an algorithm for the next-to-shortest path problem, we can solve the shortest path problem with the same time complexity since the above reduction is linear time.

In this paper, we show that the next-to-shortest path problem can be solved with the same time complexity as the single source shortest paths problem. 
Interesting future works include the directed version and the undirected case with nonnegative edge weights. 
Allowing zero-length edges makes the next-to-shortest path problem more involved. 
Particularly, Lemma~\ref{suff} no more holds, and there seems no obvious modification to generalize the lemma to the nonnegative case.  
Another open problem is if we can find the single source all destinations next-to-shortest paths in the same time complexity.

\begin{acknowledgements}
This work was supported in part by 
NSC 97-2221-E-194-064-MY3 and NSC 98-2221-E-194-027-MY3 from the National Science Council, Taiwan.
\end{acknowledgements}

\end{document}